\documentclass[preprint,12pt]{revtex4-1}
\usepackage[english]{babel}
\usepackage[latin1]{inputenc}
\usepackage{amssymb}
\usepackage{amsmath}
\usepackage{amsthm}
\usepackage{multirow}
\usepackage{url}
\usepackage{hyperref}
\usepackage{bbm}
\usepackage{dcolumn}
\usepackage{multirow}
\usepackage{graphicx}

\newcolumntype{.}{D{.}{.}{-1}}
\newcommand{\1}{\mathbbm{1}}
\makeatletter
\theoremstyle{plain}
\newtheorem{thm}{\protect\theoremname}
\theoremstyle{definition}
\newtheorem{defn}[thm]{\protect\definitionname}
\theoremstyle{remark}
\newtheorem{rem}[thm]{\protect\remarkname}
\theoremstyle{plain}
\newtheorem{lem}[thm]{\protect\lemmaname}
\ifx\proof\undefined
\newenvironment{proof}[1][\protect\proofname]{\par
\normalfont\topsep6\p@\@plus6\p@\relax
\trivlist
\itemindent\parindent
\item[\hskip\labelsep
\scshape
#1]\ignorespaces
}{%
\endtrivlist\@endpefalse
}
\providecommand{\proofname}{Proof}
\fi
\theoremstyle{plain}
\newtheorem{cor}[thm]{\protect\corollaryname}

\makeatother

\providecommand{\corollaryname}{Corollary}
\providecommand{\definitionname}{Definition}
\providecommand{\lemmaname}{Lemma}
\providecommand{\remarkname}{Remark}
\providecommand{\theoremname}{Theorem}

\begin{document}


\title{On the performance of a cavity method based algorithm for the
Prize-Collecting Steiner Tree Problem on graphs}

\author{Indaco Biazzo}
\email{indaco.biazzo@polito.it}
\affiliation{Politecnico di Torino, Corso Duca degli Abruzzi 24, 10129 Torino,
Italy}
\author{Alfredo Braunstein}
\email{alfredo.braunstein@polito.it}
\affiliation{Politecnico di Torino, Corso Duca degli Abruzzi 24, 10129 Torino,
Italy}
\affiliation{Human Genetics Foundation, Via Nizza 52, 10023 Torino, Italy}
\affiliation{Collegio Carlo Alberto, Via Real Collegio 30, 10024 Moncalieri, Italy}
\author{Riccardo Zecchina}
\email{riccardo.zecchina@polito.it}
\affiliation{Politecnico di Torino, Corso Duca degli Abruzzi 24, 10129 Torino,
Italy}
\affiliation{Human Genetics Foundation, Via Nizza 52, 10023 Torino, Italy}
\affiliation{Collegio Carlo Alberto, Via Real Collegio 30, 10024 Moncalieri,
Italy}

\date{}

\begin{abstract}

We study the behavior of an algorithm derived from the cavity method for the 
Prize-Collecting Steiner Tree (PCST) problem on graphs. 
The algorithm is based on the zero temperature limit of the cavity equations and as such is
formally  simple (a fixed point equation resolved by iteration) and distributed (parallelizable).
We provide a detailed comparison  with state-of-the-art algorithms on a wide  range of existing 
benchmarks networks and random graphs.
Specifically, we consider  an enhanced derivative of the Goemans-Williamson
heuristics and the DHEA solver, a Branch and Cut Linear/Integer Programming
based approach. 
The comparison shows that the cavity algorithm outperforms the two  algorithms
in most large instances both  in running time and quality of the solution. 
Finally we prove a few optimality properties of the solutions provided by our algorithm, including optimality under the two post-processing procedures defined in the Goemans-Williamson derivative and global optimality in some limit cases.

\end{abstract}


\maketitle

\section{Introduction}

The cavity method developed for the study of disordered systems in statistical physics
has led in the recent years to the design of a family  of algorithmic techniques  for combinatorial optimization
known as message-passing algorithms (MPA).
In spite of the numerical evidence of great potentialities of these techniques in terms of efficiency  
and quality of results for many optimization problems, their use in real-world problems has still to be fully expressed. The main reasons for this
reside in  the fact that  the  derivation of the equations underlying the algorithms
are in many cases  non-trivial and that  the  rigorous and numerical analyses of the cavity equations  are still largely incomplete.
Both rigorous results and benchmarking would play an important role in helping the process of integrating MPAs with the existing techniques.

In what follows we focus on a very well known NP-hard optimization problem over
networks, the so-called Prize Collecting Steiner Tree problem on graphs (PCST).
The PCST problem can be stated in general terms as the problem of finding a
connected subgraph of minimum cost. It has applications in many areas ranging
from biology, e.g. finding protein associations in cell signaling
\cite{fraenkel-pcst,pnas:2011}, to network technologies, e.g. finding optimal ways to deploy fiber optic and heating
networks for households and industries \cite{Hackner:2004}.

Though the cavity equations have been developed for the study of mean field models for disordered systems, the range of their applicability  is known to go beyond these problems.

In this paper we show how MSGSTEINER  -- an algorithm derived from the zero temperature cavity equations \cite{pnas:2011} -- compares
with state-of-the-art techniques on  benchmarks problem instances. Specifically, we provide comparison results
with an enhanced derivative of the Goemans-Williamson heuristics (MGW)
\cite{GW,Johnson2000} and with the DHEA solver \cite{ivanaConf}, a Branch and
Cut Linear/Integer Programming based approach. We made the comparison both on
random networks and in known benchmarks. We show that MSGSTEINER typically 
outperforms the state-of-the-art algorithms in the largest instances of the PCST  problem both in the values of the optimum and in running time. 

Finally, we show how some aspects of the solutions can be provably
characterized. Specifically we show some optimality properties of the fixed
points of the cavity equations, including optimality under the two
post-processing procedures defined in MGW (namely \textit{Strong Pruning} and
\textit{Minimum Spanning Tree}) and global optimality of the MPA solution in
some limit cases.

\subsection{Related work}
The method and the algorithm described here are a generalization of the technique presented in ref. \cite{BBZ2008a}. In \cite{BBZ2008a} the algorithm is tested on different  families of random graphs for the more specific case of bounded depth ($D$) Steiner tree problem, which can be recovered from the PCST  problem by sending to infinity the weights of the so-called terminal nodes. 
In the cases of Erdos-Renyi random graphs and for scale-free graphs the numerical performance of the algorithm have been  shown to  be extremely good  though there exits no  rigorous results to compare with. 
Interestingly enough the case of complete  graphs with random weights allows for a comparison with rigorous asymptotic results. The scaling coefficients  of the power law for the average minimum cost and number of Steiner nodes as a function of the size $N$ of the graph was calculated exactly in ref. \cite{AFW} , where it was also rigorously established that the critical depth for the bounded-depth  Minimum Spanning Tree and Steiner Tree on random complete graphs is $D =\log_2 \log N$.  
Extensive numerical studies  up top $N=10^5$ which for brevity we do not report in detail,  show that the cavity approach provides solutions which have a minimum cost that is below that of the greedy algorithm analyzed in \cite{AFW} and that there is slow convergence to the exact scaling parameters. This fact corroborates  the conjecture that the cavity approach could be asymptotically exact and reproduce the results of  \cite{AFW}. 
While this is not totally unexpected  for statistical physics of random systems (the cavity approach is known to  be very accurate on mean-field problems defined over  complete graphs), it is important for the rigorous foundation of the cavity method itself. There exist in fact very few model problems on which the zero temperature cavity approach can be proven to be exact, one famous example being the matching problem \cite{Aldous}. NP-complete problems
(considered in their typical realizations) are particularly elusive in this respect, possibly due to the local nature of the cavity algorithms. Therefore, having at hand a non-trivial problem which can be analyzed rigorously  as in  \cite{AFW} constitutes an interesting  case also for  the rigorous understanding of the cavity method.

\section{The problem: prize collecting steiner trees}

In the following we will describe the Prize-Collecting Steiner Tree problem on
Graphs (see e.g. \cite{Johnson2000,Lucena2004}).
\begin{defn}
Given a network $G=(V,E)$ with positive (real) weights $\{c_e:e\in E\}$ on edges
and $\{b_i:i\in V\}$ on vertices, consider the problem of finding the connected
sub-graph $G'=(V',E')$ that minimizes $H(V',E')=\sum_{e\in
E'}c_{e}-\lambda\sum_{i\in V'}b_{i}$, i.e. to compute the minimum:
\begin{equation}
\min_{\begin{array}{c}
E'\subseteq E,V'\subseteq V\\
(V',E')\,\,\mbox{connected}\end{array}}\sum_{e\in E'}c_{e}-\lambda\sum_{i\in
V'}b_{i}.
\label{eq:H}
\end{equation}
\end{defn}

It can be easily seen that a minimizing sub-graph must be a tree (links closing
cycles can be removed, lowering $H$). The parameter $\lambda$ regulates the
tradeoff between the edge costs and vertices prizes, and its value has the
effect to determine the size of the subgraph $G'$: for $\lambda=0$ the empty
subgraph is optimal, whereas for $\lambda$ large enough the optimal subgraph
includes all nodes.

This problem is known to be NP-hard, implying that no polynomial algorithm exists 
that can solve any instance of the problem unless $NP=P$. To solve it we will
use a variation of a very efficient heuristics based on belief propagation
developed on \cite{BBZ2008a} that is known to be exact on some limit cases
\cite{BBZ2008a,BBZ2008b}. We will partially extend the results in
\cite{BBZ2008b} to a more general PCST setting.

\subsection{Rooted, depth bounded PCST and forests}

We will deal with a variant of the PCST called $D$-bounded rooted PCST
($D$-PCST). This problem is defined by a graph $G$, an edge cost matrix $c$ and
prize vector $b$ along with a selected ``root'' node $r$. The goal is to find
the $r$-rooted tree with maximum depth $D$ of minimum cost, where the cost is
defined as in (\ref{eq:H}). A general PCST can be reduced to $D$-bounded rooted
PCST by setting $D=|V|$ and probing with all possible rootings, slowing the
computation by a factor $|V|$ (we will see later a more efficient way of doing
it). A second variant which we will consider is the so-called  $R$ multi-rooted $D$-bounded Prize Collecting Steiner Forest
($(R,D)$-PCSF). It consists of is a natural generalization of the previous problem: a
subset $R$ of ``root'' vertices is selected, and the scope is to find a forest
of trees of minimum cost, each one rooted in one of the preselected root nodes
in $R$.

\subsection{Local constraints}

The cavity formalism can be adopted and made efficient if the global constraints  which may be present in the problem can be written 
in terms of local constraints. In the  PCST case the global constraint is connectivity which can be made local as follows.

We start with the graph $G=\left(V,E\right)$ and a selected \emph{root} node
$r\in V$. To each vertex $i\in V$ there is an associated couple
of variables $\left(p_{i},d_{i}\right)$ where $p_{i}\in \partial i\cup\left\{
*\right\}$, $\partial i=\{j:(ij)\in E\}$ denotes the set of neighbors of $i$ in
$G$ and $d_{i}\in\left\{ 1,\dots,D\right\} $. Variable $p_{i}$ has the
meaning of the parent of $i$ in the tree (the special value $p_i=*$ means that
$i\notin V'$), and $d_{i}$ is the auxiliary  variable describing  its distance
to the root node (i.e. the \emph{depth} of $i$). To correctly describe a tree,
variables $p_i$ and $d_i$ should satisfy a number of constrains, ensuring that
depth decreases along the tree in
direction to the root, i.e. $p_{i}=j\Rightarrow d_{i}=d_{j}+1$. Additionally,
nodes that do not participate to the tree ($p_i=*$) should not be parent of some
other node, i.e. $p_i=j\Rightarrow p_j\neq*$. 
Note
that even though $d_{i}$ variables are redundant (in the sense that
they can be easily computed from $p_{j}$ ones), they are crucial to maintain
the locality of the constraints. For every ordered couple $i,j$ such that
$(ij)\in E$, we define
$f_{ij}\left(p_{i},d_{i},p_{j},d_{j}\right)=\1_{p_{i}=j\Rightarrow
d_{i}=d_{j}+1\wedge p_j\neq
*}=1-\delta_{p_{i},j}\left(1-\delta_{d_{i},d_{j}+1}(1-\delta_{p_j,*})\right)$
(here $\delta$ is the Kroenecker delta). The condition of the subgraph to be a
tree can be ensured by imposing that 
$g_{ij}=f_{ij} f_{ji}$ has to be equal to one for each edge $\left(ij\right)\in
E$. 
If we extend the definition of $c_{ij}$ by $c_{i*}=\lambda b_i$, then (except
for an irrelevant constant additive term), the minimum in (\ref{eq:H}) equals
to:
\begin{equation}
\min{\left\{ \mathcal{H}(\mathbf{p}): (\mathbf{d},\mathbf{p}) \in {\mathcal
T}\right\}},
\end{equation}
where $\mathbf{d}=\{d_i\}_{i\in V}$, $\mathbf{p}=\{p_i\}_{i\in V}$, $\mathcal T
= \{ (\mathbf{d},\mathbf{p}):g_{ij}(p_i,d_i,p_j,d_j)=1\,\forall (ij)\in E )$
and 
\begin{equation}
\mathcal{H}(\mathbf{p}) \equiv \sum_{i\in V}c_{ip_i}.
\end{equation}

This new expression for the energy accounts for the sum of taken edge costs plus
the sum of uncollected prizes and has the advantage of being non-negative.

\section{Derivation of the message-passing cavity equations}

The algorithmic scheme we propose originates from the cavity method of
statistical physics, a technique  which is known in other fields under 
different names, namely Cavity equations, Belief Propagation (BP), Max-Sum or
Sum-Product equations (MS).  From a numerical point of view, message-passing
algorithms are distributed algorithm which allow  for a very fast
resolution of inference and optimization problems \cite{Mezard:2002lo}, 
even for large networks. 
A recent review can be found in \cite{mezard-montanari}. The starting point for the equations  is
the Boltzmann-Gibbs distribution:
\begin{equation}
P(\mathbf{d},\mathbf{p})= \frac{\exp (-\beta \mathcal{H}(\mathbf{p}))}{Z_\beta},
\end{equation}
where $(\mathbf{d},\mathbf{p})\in \mathcal T$, $\beta$ is a positive parameter
(called inverse temperature), and $Z_\beta$ is a normalization constant (called
partition function). In the limit $\beta \to \infty$ this probability
concentrates on the configurations which minimize $\mathcal{H}$. 
The BP approximation consists in a weak correlation assumption between certain
probability distributions of single $(p_i,d_i)$ pairs called ``cavity marginals''. Given
$i,j\in V$, the cavity marginal $P_{j   i} \left(d_{j},p_{j}\right)$ is defined
as the marginal distribution 
$\sum_{(d_k,p_k)_{k\in V\setminus\{j,i\}}} P_{G^{(i)}}(\mathbf{d}, \mathbf{p})$ on
a graph $G^{(i)}$ from which node $i$ has been temporally removed. The BP
equations are derived by assuming that the  cavity marginals are uncorrelated and as such satisfy the following  closed set of
equations (see e.g.~\cite{mezard-montanari} for a general discussion):

\begin{eqnarray}
{P}_{j  i}\left(d_{j},p_{j}\right) & \propto & e^{-\beta
c_{jp_{j}}}\prod_{k\in\partial j\setminus i}Q_{k 
j}\left(d_{j},p_{j}\right)\label{eq:phat}\\
Q_{k  j}\left(d_{j},p_{j}\right) & \propto & \sum_{d_{k}}\sum_{p_{k}}P_{k 
j}\left(d_{k},p_{k}\right)g_{jk}\left(d_{k},p_{k},d_{j},p_{j}\right).
\label{eq:bp}
\end{eqnarray}

\begin{figure}
 \begin{center}
\includegraphics[width=0.4\textwidth]{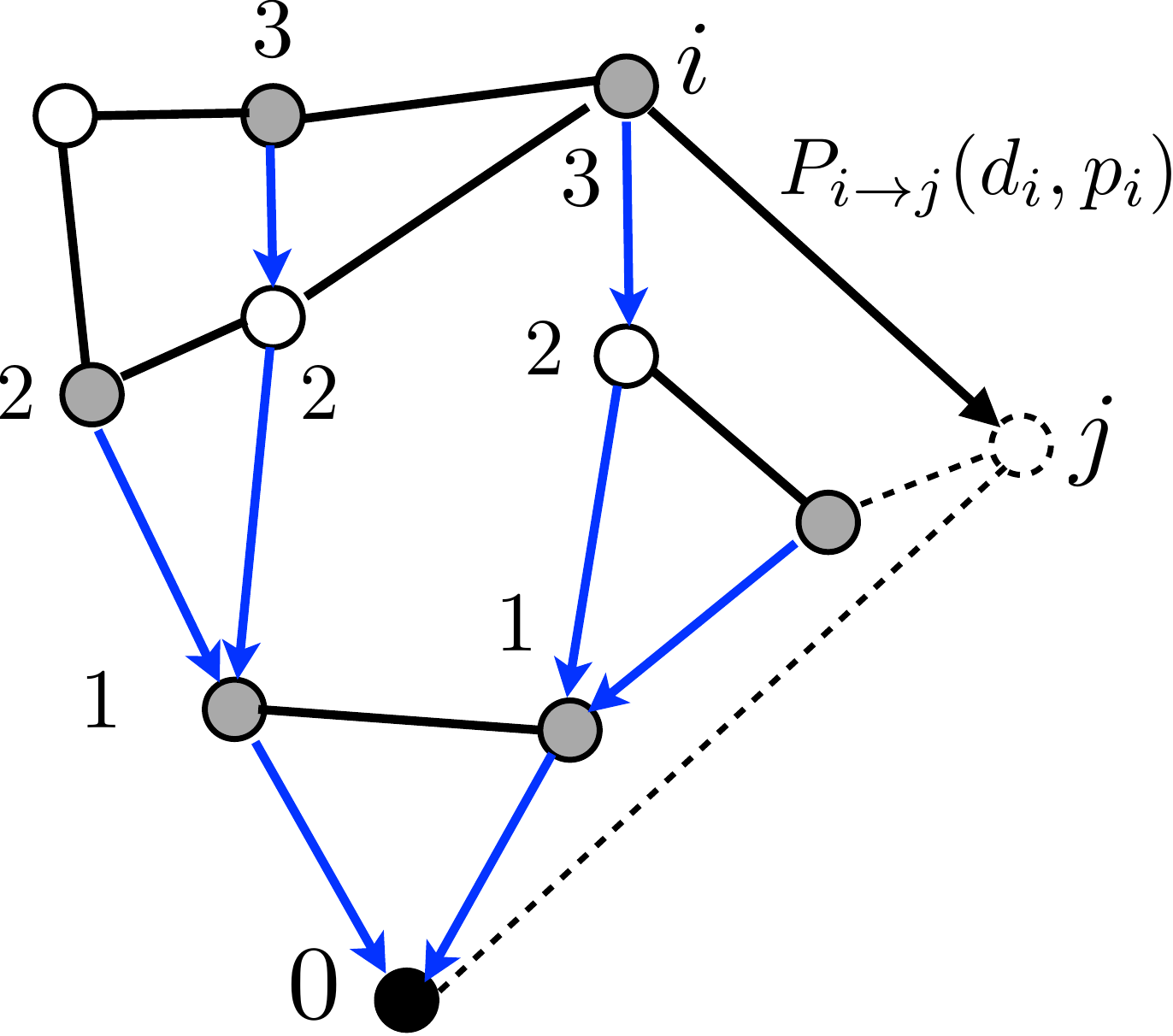}
  \caption{A schematic representation of the Prize Collecting Steiner Tree problem and its local representation.
  Numbers next to the nodes are the distances (depths) from the root node (black node). The prize value is proportional to the darkness of the nodes.  
  Arrows are the pointers from node  to node. Distances and pointers are used to define the  connectivity constraints
  which appear in the message-passing equations.|
  Blue arrows represent a potential solution.}
  \label{fig_stener}
 \end{center}
 \end{figure}

This assumption is correct if $G$ is a tree, in which case  
(\ref{eq:phat})-(\ref{eq:bp}) are exact and have a unique solution.
Equations (\ref{eq:phat})-(\ref{eq:bp}) can be seen as fixed point equations,
and solutions are normally searched through iteration: substituting \eqref{eq:bp} onto \ref{eq:phat} and giving a time index $t+1$
and $t$ to the cavity marginals in respectively the left and right hand side of the resulting equation, this system is iterated until numerical
convergence is reached. Cavity marginals are often called ``messages'' because they can be
thought of as bits of information that flow between edges of the graph during
time in this iteration. On a fixed point, the BP approximation to the marginal is
computed as
\begin{eqnarray}
{P}_{j}\left(d_{j},p_{j}\right) & \propto & e^{-\beta c_{jp_{j}}}\prod_{k\in\partial j}Q_{kj}\left(d_{j},p_{j}\right)\label{eq:bp-marginas}.
\end{eqnarray}
\subsection{Max-sum: $\beta \to \infty$ limit}

In order to take the $\beta \to \infty$ limit, (\ref{eq:bp}) can be rewritten in
terms of ``cavity fields''
\begin{eqnarray}
\psi_{j  i}\left(d_{j},p_{j}\right) & = & \beta^{-1}\log P_{j 
i}\left(d_{j},p_{j}\right)\\
\phi_{k  j}\left(d_{j},p_{j}\right) & = & \beta^{-1}\log Q_{k 
j}\left(d_{j},p_{j}\right).
\end{eqnarray}
The BP equations  take the so-called MS form:
\begin{eqnarray}
\label{eq:ms1}
\psi_{j  i}\left(d_{j},p_{j}\right) & = & -c_{jp_{j}}+\sum_{k\in\partial
j\setminus i}\phi_{k  j}\left(d_{j},p_{j}\right)+C_{ji}\label{eq:psi}\\
\label{eq:ms2}
\phi_{k  j}\left(d_{j},p_{j}\right) & = &
\max_{p_{k},d_{k}:g_{jk}\left(d_{k},p_{k},d_{j},p_{j}\right)=1}\psi_{k 
j}\left(d_{k},p_{k}\right)\label{eq:max-sum},
\end{eqnarray}
where $C_{ji}$ is an additive constant chosen to ensure $\max_{d_j,p_j}\psi_{j 
i}\left(d_{j},p_{j}\right)=0$
 
Computing the right side of (\ref{eq:max-sum}) is in general too costly in
computational terms. Fortunately, the computation can be carried out efficiently by
breaking up the set over which the max is computed into smaller (possibly
overlapping) subsets. We define
\begin{eqnarray}
A_{k  j}^{d} & = & \max_{p_{k}\neq j,*}\psi_{k 
j}\left(d,p_{k}\right)\label{eq:A}\\
B_{k  j}^{d} & = & \psi_{k  j}\left(d,*\right)\\
C_{k  j}^{d} & = & \psi_{k  j}\left(d,j\right)\label{eq:C}.
\end{eqnarray}

Equation (\ref{eq:max-sum}) can now be rewritten as: 
\begin{eqnarray}
A_{j  i}^{d} & = & \sum_{k\in\partial  j\setminus i}E_{k 
j}^{d}+\max_{k\in\partial i\setminus j}\left\{ -c_{jk}-E_{k  j}^{d}+A_{k 
j}^{d-1}\right\} \label{eq:AFP}\\
B_{j  i} & = & -c_{j*}+\sum_{k\in\partial j\setminus i}D_{k  j}\\
C_{j  i}^{d} & = & -c_{ji}+\sum_{k\in\partial j\setminus i}E_{k  j}^{d}\\
D_{j  i} & = & \max\left(\max_{d}A_{j  i}^{d},B_{j  i}\right)\\
E_{j  i}^{d} & = & \max\left(C_{j  i}^{d+1},D_{j  i}\right)\label{eq:EFP}.
\end{eqnarray}

Using some simple efficiency tricks including computing $\sum_{k\in\partial
j\setminus i}E_{kj}^{d}$ as $\sum_{k\in\partial j}E_{kj}^{d} - E_{ki}^{d}$, the
computation of the right side of (\ref{eq:AFP})-(\ref{eq:EFP}) for all $i\in
\partial j$ can be done in a  time proportional to $D|\partial j|$, where $D$
is the depth bound. The overall computation time is then $O(|E| D)$ per
iteration.

\subsection{Total fields}

In order to identify the minimum cost configurations, we need to compute the
total marginals, i.e. the marginals in the case in which  no node has been
removed from the graph. Given cavity fields, the total fields
$\psi_{j}\left(d_{j},p_{j}\right)=\lim_{\beta\to\infty}\beta^{-1}\log
P_{j}\left(d_{j},p_{j}\right)$ can be written as:
\begin{eqnarray}
\psi_{j}\left(d_{j},p_{j}\right) & = & -c_{jp_{j}}+\sum_{k\in\partial j}\phi_{k 
j}\left(d_{j},p_{j}\right) + C_j\label{eq:marginal},
\end{eqnarray}
where $C_j$ is again an additive constant that ensures 
$\max_{d_{j},p_{j}}\psi_{j}\left(d_{j},p_{j}\right)=0$.
In terms of the above quantities we find  $\psi_{j}\left(d_{j},i\right)=F_{j 
i}^{d}\overset{def}{=}\sum_{k\in\partial j}E_{k  j}^{d}+\left(-c_{ij}-E_{j 
i}^{d}+A_{j  i}^{d-1}\right)$
if $i\in\partial j$ and
$\psi_{j}\left(d_{j},*\right)=G_{j}\overset{def}{=}-c_{j*}+\sum_{k\in\partial
j}D_{k  j}$.
The total fields can be interpreted as (the Max-Sum approximation to)
the relative negative energy loss of chosing a
given configuration for variables $p_j,d_j$ instead of their optimal choice, 
i.e. $\psi_{j}\left(d_{j},p_{j}\right) = \min{\left\{ \mathcal{H}(\mathbf{p}'):
(\mathbf{d}',\mathbf{p}') \in {\mathcal T}\right\}}-\min{\left\{
\mathcal{H}(\mathbf{p}'): (\mathbf{d}',\mathbf{p}') \in {\mathcal T},
d_j=d'_j,p_j=p'_j\right\}}$. In particular, in absence of degeneracy, the maximum of the field is attained for  values of $p_j, d_j$ corresponding to the optimal energy.

\subsection{Iterative dynamics and reinforcement}

Equations (\ref{eq:AFP})-(\ref{eq:EFP}) can be thought as a fixed-point equation
in a high dimensional euclidean space. This equation could be solved by repeated
iteration of the quantities $A$,$B$, and $C$ starting from an arbitrary initial
condition, simply by adding an index $(t+1)$ to $A,B,C$ in the left-hand side of
(\ref{eq:max-sum}) and index $(t)$ to all other instances of $A,B,C,D,E$.


This system converges in many cases. When it does not converge,
a technique called \emph{reinforcement} is of help \cite{BrZe2006b}. 
The idea is to
perturbate the right side of (\ref{eq:psi}) and (\ref{eq:marginal})
by adding the term $\gamma_{t}\psi_{j}^{t}\left(d_{j},p_{j}\right)$
for a (generally small) scalar factor $\gamma_{t}$. The resulting
equations become:

\begin{eqnarray}
A_{j  i}^{d}\left(t+1\right) & = & \sum_{k\in \partial j\setminus i}E_{k 
j}^{d}\left(t\right)+\\
& & + \max_{k\in \partial j\setminus i}\left\{ -c_{jk}-E_{k 
j}^{d}\left(t\right)+A_{k  j}^{d-1}\left(t\right)+\gamma_{t}F_{j 
k}^{d}\left(t\right)\right\} \\
B_{j  i}\left(t+1\right) & = & -c_{j*}+\sum_{k\in \partial j\setminus i}D_{k 
j}\left(t\right)+\gamma_{t}G_{j}\left(t\right)\\
C_{j  i}^{d}\left(t+1\right) & = & -c_{ji}+\sum_{k\in \partial j\setminus i}E_{k
 j}^{d}\left(t\right)+\gamma_{t}F_{j  i}^{d}\left(t\right)\\
D_{j  i}\left(t\right) & = & \max\left\{\max_{d}A_{j  i}^{d}\left(t\right),B_{j 
i}\left(t\right)\right\}\\
E_{j  i}^{d}\left(t\right) & = & \max\left\{C_{j  i}^{d+1}\left(t\right),D_{j 
i}\left(t\right)\right\}\\
G_{j}\left(t+1\right) & = & -c_{j*}+\sum_{k\in \partial j}D_{k 
j}\left(t\right)+\gamma_{t}G_{j}\left(t\right)\\
F_{j  i}^{d}\left(t+1\right) & = & \sum_{k\in \partial j}E_{k 
j}^{d}\left(t\right)+\left(-c_{ji}-E_{i  j}^{d}\left(t\right)+A_{i 
j}^{d-1}\left(t\right)\right)+\\
& & +\gamma_{t}F_{j  i}^{d}\left(t\right).
\end{eqnarray}

In our experiments, the equations converge for a sufficiently 
large $\gamma_t$. The strategy we adopted is, when the equations do 
not converge, to start with $\gamma_t=0$ 
and slowly increase it until convergence in a linear regime
$\gamma_t = t\rho$ (although other regimes are possible). The number of 
iterations is then found to be inversely dependent on the parameter $\rho$. This strategy 
could be interpreted as using time-averages of the MS marginals when the equations
do not converge to gradually bootstrap the system into an (easier to solve) system 
with sufficiently large external fields. A C++ implementation of these 
equations can be found (in source form) on \cite{cmpwebsite}. Note that the cost matrix $\left(c_{ij}\right)$ need not to be symmetric, and the same scheme could be used for directed graphs (using $c_{ji}=\infty$ if $\left(i,j\right)\in E$ but $\left(j,i\right)\notin E$). 

\subsection{Root choice}

The PCST formulation given in the introduction is unrooted. The MS equations on
the other hand, need a predefined root. One way of reducing the unrooted problem
to a rooted problem is to solve $N=|V|$ different problems with all possible
different rooting, and choose the one of minimum cost. This unfortunately adds a
factor $N$ to the time complexity. Note that in the particular case in which
some vertex has a large enough prize to be necessarily included in an optimal
solution (e.g. $\lambda b_i > \sum_{e\in E} c_e$), this node can simply be
chosen as as root. 

We have devised a more efficient method for choosing the root in the general
case, which we will now describe. Add an extra new node $r$ to the graph,
connected to every other node with identical edge cost $\mu$. If $\mu$ is
sufficiently large, the best energy solution is the (trivial) tree consisting in
just the node $r$. Fortunately, a solution of the MS equations on this graph
gives additional information: for each node $j$ in the original graph, the marginal 
field $\psi_j$ gives the relative energy shift of selecting a given parent 
(and then adjusting all other variables in the best possible configuration). 
Now for each $j$, consider the positive real value $\alpha_j = -\psi_j(1, r)$, that corresponds with the best attainable energy, constrained to the condition that $r$ is the parent of $j$.
If $\mu$ is large enough, this energy is the energy of a tree in which only $j$ 
(and no other node) is connected to $r$ (as each of these connections costs $\mu$).
But these trees are in one to one correspondence with trees rooted at $j$ in the 
original graph. The smallest $\alpha_j$ will thus identify an optimal rooting.

Unfortunately the information carried by these fields is not sufficient to build the optimal 
tree. Therefore  one needs to select the best root $j$ and run the MS equations 
a second time on the original graph using this choice.

\subsection{Comparision with other techniques}

We compared the performance of MSGSTEINER with the three different algorithms: two that employ
an integer linear programming strategy to find an optimal subtree, namely the Lagrangian
Non Delayed Relax and Cut (LNDRC) \cite{Cuna} and branch-and-cut (DHEA)
\cite{ivanaConf}, and modified version of the Goemans and Williamson algorithm (MGW)\cite{Johnson2000}.

\subsubsection{Integer Linear programming}
 
The goal of the Integer linear programming (ILP) is to find an integer solution
vector $x^*\in \mathbb{Z}^n$ such that:
 \begin{equation}
c^T x^*= \min \{c^T x^*\: | \: Ax\geq b,\: x\in   \mathbb{Z}^n\},
\label{eq:ilp}
\end{equation}
where a matrix $A\in \mathbb{R}^{m*n} $ and vector $b \in \mathbb{R}^m$ and $c
\in \mathbb{R}^n$ are given. Many graph problems can be formulated as an integer
linear programming problem \cite{Aardal:1996}. In general, solving
(\ref{eq:ilp}) with $x^*\in\mathbb{Z}$ is NP-Complete. The standard approach
consists in solving (\ref{eq:ilp}) for $x^* \in \mathbb R$ (a relaxation of the
original problem) and use the solution as a guide for some heuristics or
complete algorithm for the integer case. The relaxed problem can be solved by
many classical algorithms, like the Simplex Method \cite{dantziglinear:2003} or
Interior Point methods.
In order to map the PCST problem in a ILP problem we introduce a variable vector
$z \in \{ 0,\, 1\}^{E}$ and $y \in \{ 0,\, 1\}^{V}$ where the component for an
edge in $E$ or for a vertex in $V$ is one if and only if it is included in the
solution, zero otherwise. Now (\ref{eq:H}) can be written as 
\begin{equation}
H= \sum_{e \in E} c_e z_e - \sum_{i \in V} b_i y_i \:,
\label{eq:ilp_cost}
\end{equation} 
and the constraints $Ax\geq b$ in (\ref{eq:ilp}) generally involve all the
variable $z$ and $y$ and describe the problem. For the PCST, and in general for
hard problems, the number of constraints grows
exponentially with the problem size\cite{Aardal:1996}. DHEA and LNRDC use
different techniques to tackle the problems of enormous number of constraints.
Both programs are able in principle to prove the optimality of solution, and
when is not the case they are able to give a lower bound for the value of the
optimum.

\subsubsection{Goemans-Williamson}
The MGW algorithm is based on the primal-dual method for approximation
algorithms \cite{GW}. The starting point is still the ILP formulation of the
problem (\ref{eq:ilp}), but it employs a controlled approximation scheme that
enforces the cost of any solution to be at most twice as large as the optimum
one. In addition, MGW implements two different post-processing strategies,
namely a pruning scheme that is able to eliminate some nodes while lowering the
cost, and the computation of the minimum spanning tree in order to find an
optimal rewiring of the same set of nodes. The overall running time is
$\mathcal{O}{(n^2\log n)}$.  A complete description is available in \cite{GW}.

\section{Computational Experiments}

\subsection{Instances}
Experiments were performed on several classes of instances: 

\begin{itemize}
\item  C, D and E available at \cite{IvanaLjubicSite} and derived from the
Steiner problem instances of the OR-Library \cite{OR_Library}. This set of 120
instances was previously used as benchmark for algorithms for the PCST\cite{OR_Library}. 
The solutions of these instances were obtained with the algorithms\cite{ivanaConf, Cuna}. The
class C, D, E have respectively $500$, $1000$, $2000$ node and are generated at
random, with average vertex degree is either $2.5$, $4$, $10$ or $50$. Every edge cost
is a random integer in the interval $[1,\,10]$. There are either  $5$, $10$,
$n/6$, $n/4$ or $n/2$ vertices with prizes different from zero and random
integer generated in the interval $[1,\, maxprize]$ where $maxprize$ is either
$10$ or $100$. Thus, each of the classes C, D, E consists of 40 graphs.  

\item K and P available at \cite{IvanaLjubicSite}. These instances are provided
in \cite{Johnson2000}. In the first group instances are unstructured. The second
group includes random geometric instances designed to have a structure somewhat
similar to street maps. Also the solution of these instances were found with the
algorithms \cite{ivanaConf, Cuna}.

\item H are the so-called hypercubes instances proposed in \cite{Rossetti}. Sets
of artificially generated and very difficult instances for the Steiner tree
problem. Graphs  are d-dimensional hypercubes with $d \in {6, . . . , 12}$. For
each value of $d$, the corresponding graph has $2^d$ vertices and $d \cdot 2^{d
- 1}$ edges. We used the prized version of these instances defined in \cite{ivanaConf}. 
For almost all instances in this class the optimum is unknown.

\item i640 are the so-called incidence instances proposed in \cite{DuinVob} for
the Minimum Steiner Tree problem. These instances have $640$ nodes and only the nodes in
a subset $K\subseteq V$ have prizes different from zero (in the original problem
these were terminals). The weight on each edge $(i,\, j)$ is defined with a
sample $r$ from a normal distribution, rounded to an integer value with a
minimum outcome of $1$ and maximum outcome of $500$, i.e., $c_{ij} = \min \{
\max \{ 1,\, round(r) \},\, 500 \} $. However, to obtain a graph that is much
harder to reduce by preprocessing techniques three distributions with a
different mean value are used. Any edge $(i, j)$ is incident to none, to one, or
to two vertices in subset $K$. The mean of r is $100$ for edges $(i, j)$ with
$i,\, j \notin K$, $200$ on edges with one end vertex in $K$, and 300 on edges
with both ends in $K$. Standard deviation for each of the three normal
distributions is $5$. In order to have prizes also on vertices we extracted
uniformly from all integer in the interval between $0$ and $4*max_{edge}$ where
$max_{edge}$ is the maximum value of edges in the samples considered.  There are
20 variants combining four different number of vertices in $K$ (rounding to the
integer value $[.]$): $|k| = [\log_2{|V|} ],\, [ \surd{|V|} ],\, [2\surd{|V|}]$,
and $[ |V| /4]$ with five edge number: $|E|= [3|V|/2],\, 2|V|,\, [ |V| \log{|V|}
],\, [2|V| \log{|V|} ]$, and $[ |V| ( |V| - 1)/4]$. Each variant is drawn five
times, giving 100 instances.

\item Class R. The last class of samples are $G(n,p)$ random graphs with $n$ vertices and
independent edge probability $p= (2 \nu)/(n-1)$.
The parameter $\nu$ is the average node degree, that was chosen as $\nu=8$. 
The weight on each edge $(i,\, j)$ can take three different
value, $1,\,2$ and $4$, with equal probability $\frac{1}{3}$. Node prizes were
extracted uniformly in the interval $[ 0,\, 1 ]$. We generated different
graphs with four different values of $\lambda$ $( \lambda =
1.2,\, 1.5, \, 2 $ or $ 3 ) $, see (\ref{eq:H}),  in order to explore different regimes of solution 
sizes.
We find that the average number of nodes that belong to the solution for $ \lambda = 1.2,\,
1.5, \, 2 $ and $ 3 $ are respectively about $ 14 \% , \,  33 \% , \, 51 \% , \,
67 \% $ of the total nodes in the graph. We have created twelve instances of different
size for the four class of random graph, from $n=200$ up to $n=4000$ nodes.
For each parameter set we generated ten different realizations.
The total number of samples is 480.

\end{itemize}
The MSGSTEINER algorithm was implemented in C++ and run on a single core of an
AMD Opteron Processor 6172, 2.1GHz, 8 Gb of RAM, with Linux, g++ compiler, -O3
flag activated. A C++ implementation of these equations can be found in source
form on \cite{cmpwebsite}. The executable of DHEA is available in
\cite{IvanaLjubicSite}, and in order to compare the running time we ran DHEA 
and MSGSTEINER on the same workstation. The executable of LNDRC and
MGW programs was not available. We implemented the non-rooted version of MGW to compare only the
optimum on the random graph instances.

\subsection{Results}

We analyzed two numeric quantities: the time to find the solution, and the gap
between the cost of the solution and the best known lower bound (or the optimum
solution when available) typically found with programs based on linear
programming. The gap is defined as $gap = 100 \cdot \frac{Cost-LowerBound}{lowerBound}$.

In Table (\ref{table:resCDE}) we show the comparison between MSGSTEINER and the
DHEA program. DHEA is able to solve exactly K, P and C, D, E instances. 
The worst performance of
MSGSTEINER is on the K class, where the average gap is about 2.5\%. In this
class the average solution is very small as it comprises only about 4.4\% of
total nodes of the graph. MSGSTEINER seems to have most difficulty with small
subgraphs. MSGSTEINER is able to find solutions very close to the optimum for
the P class, that should be model a street network. MSGSTEINER is also able to 
find solutions very close to the optimum, with a gap inferior to $0.025\%$ on the C, D, and E classes. 

In Figure (\ref{fig_gap}) we show the gap of MSGSTEINER and MGW from the optimum
values found by the DHEA program in the class R. MSGSTEINER gaps are almost negligible (always under $0.05\%$) and tend to zero when the size grows. MGW gaps instead are always over $1\%$. For intermediate size of solutions trees the gaps of MGW are over $3\%$.

\begin{figure}
 \begin{center}
\includegraphics[width=1.0\textwidth]{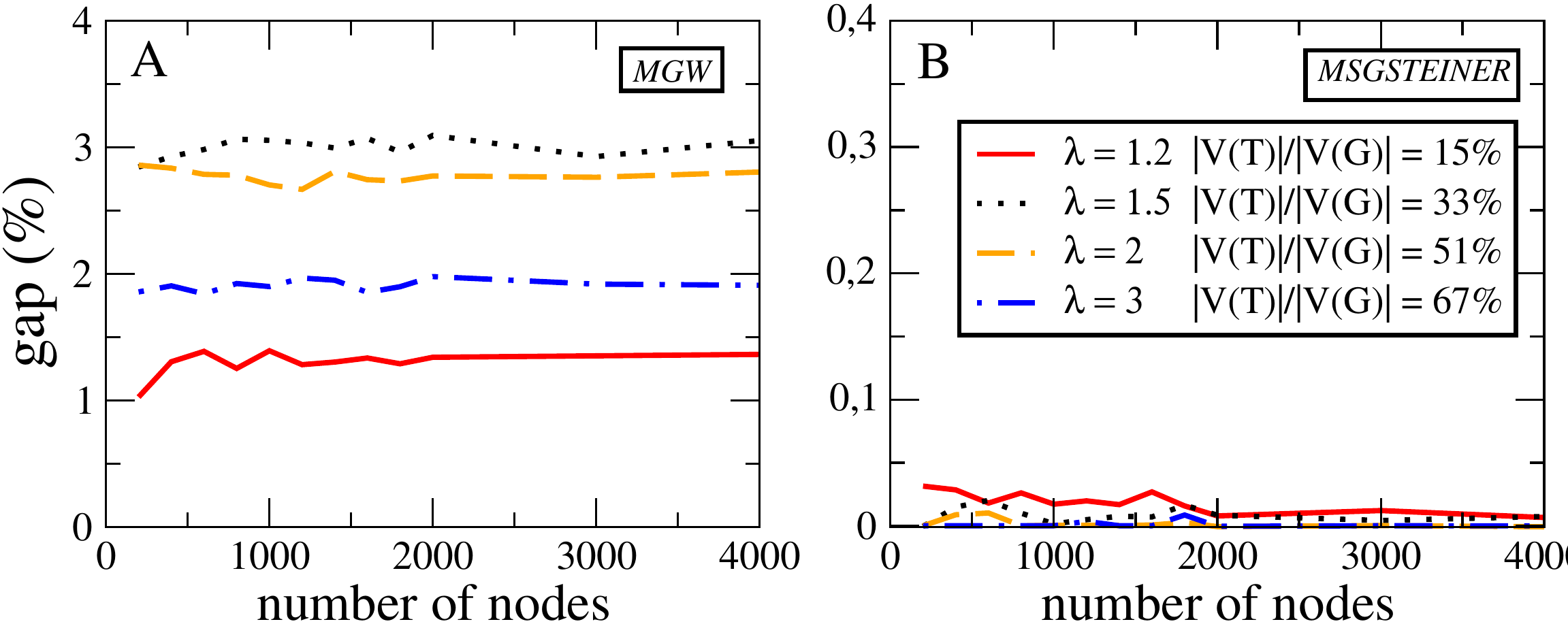}
  \caption{Plot of the Gap of the MSGSTEINER and MGW from the optimum found by
DHEA program. MSGSTEINER gaps are always under $0.05\%$. 
MGW gaps are always over $1\%$ and for intermediate sizes of the solution tree the gaps of MGW are over $3\%$. 
}\label{fig_gap}
 \end{center}
 \end{figure}

In Figure (\ref{fig_time}) we show the running time for the class R, with increasing solution tree size. In general we observe that the running time of MSGSTEINER grows much slower than the one of DHEA for increasing number of nodes in the graph and MSGSTEINER largely outperforms DHEA in computation time for large instances; furthermore the differences between the algorithms became specially and large for large expected tree solution. In at least one case DHEA could not find the optimum solution whithin the required maximum time and the MSGSTEINER solution was slightly better.

\begin{figure}
 \begin{center}
\includegraphics[width=1.0\textwidth]{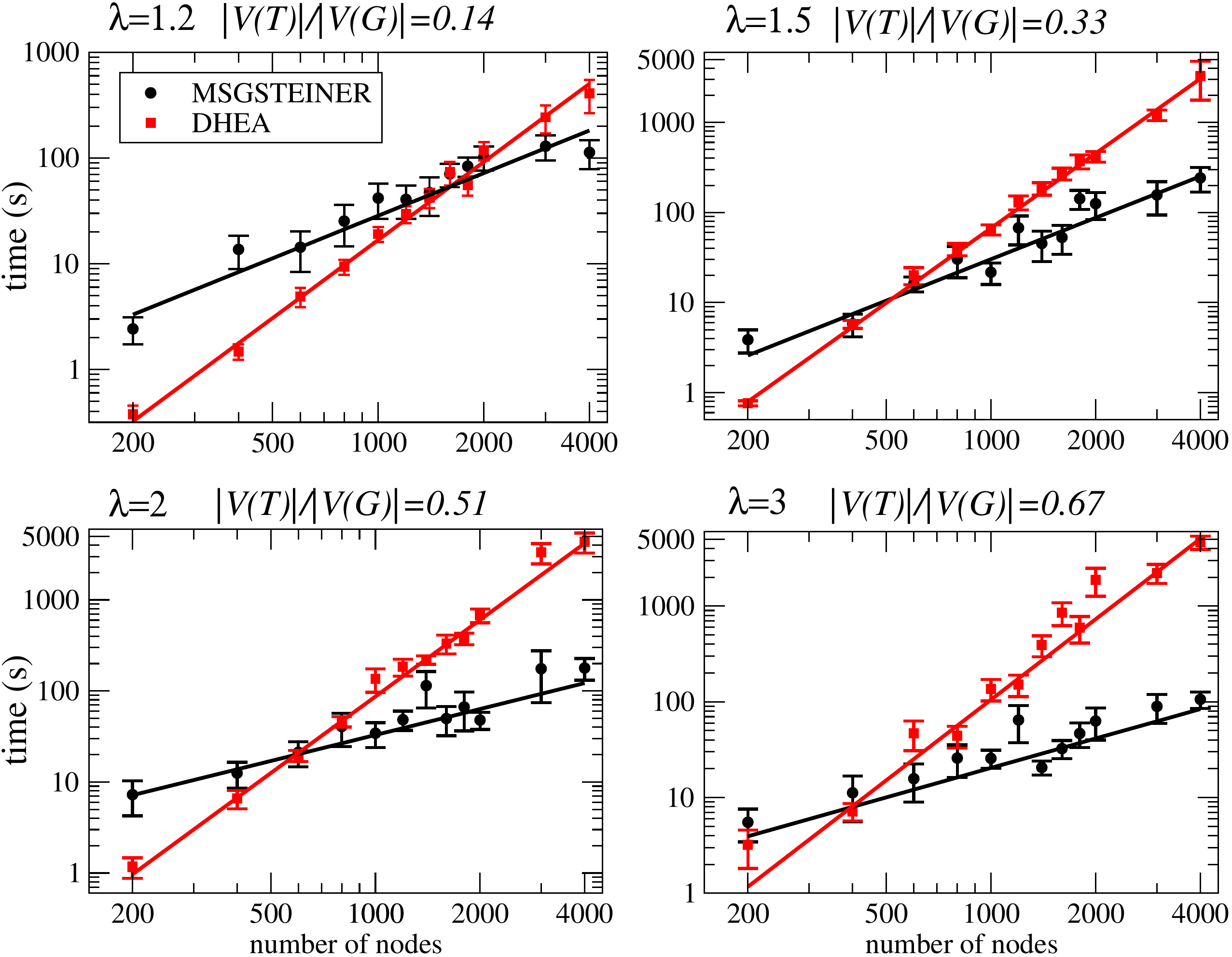}
  \caption{Result on the random graphs class R. Points correspond to the running time of MSGSTEINER and DHEA versus graph size. The four cases show how running time behavior depends on size of the expected solution tree. The quantities shown in figure are averaged over ten different realizations. Data are fitted with function 
$y = a x^{b}$. The $b$ values found for DHEA are respectively (clockwise from top-left): 2.4, 2.8, 2.8, 2.8. BP performance is as expected roughly linear in the number of vertices. The fitted $b$ parameters are (clockwise from top-left): 1.5, 1.3, 1, 1. For instances that are large enough, the running time of MSGSTEINER is smaller than the one of DHEA and the difference increases with the expected solution tree.\label{fig_time}}
 \end{center}
 \end{figure}

The class i640 consists in graphs with varying number of edges and nodes, and a
varying number of nodes with non-zero prize. We define $K$ as the subset of
nodes with non-zero prize. Table (\ref{table:resi640}) shows, for each type of
graph, the average time and the average gap on five different realizations of
the graphs for MSGSTEINER and DHEA algorithms. We set the time limit to find a
solution of DHEA to $2000$ seconds. We observe that DHEA obtains good
performance in terms of the optimality of the solution when the size of subset
$K$ is small. MSGSTEINER finds better result than DHEA when the size of $K$ is
sufficiently large, within a time of one or two order of magnitude smaller.
Moreover DHEA seems to have difficulty to find reasonable good solution when the graph
have high connectivity.

We show in Table (\ref{table:resH}) a comparison between MSGSTEINER, LNDRC
\cite{Cuna} and DHEA. The results and running time of  LNDRC are taken from
\cite{Cuna}. The computer reportedly used for the optimization is comparable
with ours. We have imposed to DHEA a time limit of 6000 seconds and we show two
results of MSGSTEINER with different values of the reinforcement parameter. The
lower bound is taken from \cite{Cuna}. In almost all instance MSGSTEINER
obtains better results, both in time and in quality of solution. The difference
is accentuated  for large instances. As expected, decreasing the reinforcement
parameter allows to find lower costs at the expense of larger computation
times.

\begin{table}
\centering  
\begin{tabular}{r.....} 
\hline                        
Group & \multicolumn{1}{c}{MS gap} & \multicolumn{1}{c}{MS time (s)} &
\multicolumn{1}{c}{DHEA gap} &\multicolumn{1}{c}{DHEA time (s)} & 
\multicolumn{1}{c}{Size Sol} \\  
\hline                  
K & 2.62 \%& 6.51  & 0.0 \%& 127.97 &  4.4 \% \\ 
P & 0.46 \%& 2.31  & 0.0 \%& 0.18 & 31.4 \% \\
C & 0.006 \%& 16.24 & 0.0 \%& 2.30  & 20.2 \% \\
D & 0.005 \%& 35.06 & 0.0 \%& 16.12 & 20.2 \% \\
E & 0.024 \%& 305.49 & 0.0 \%& 1296.11 & 26.4 \% \\ 
\hline 
\end{tabular}
\caption{Results class KPCDE} 
\label{table:resCDE} 
\end{table}

\begin{table}
\centering
\begin{tabular}{r.....}
\hline                       
Name & \multicolumn{1}{c}{time MS} & \multicolumn{1}{c}{time DHEA} &
\multicolumn{1}{c}{gap MS (\%)} & \multicolumn{1}{c}{gap DHEA (\%)} \\
\hline
0-0 & 0.8 & 0.2 & 1.3 & 0 \\
0-1 & 2.5 & 4.2 & 1.0 & 0 \\
0-2 & 100.8 & 226.6 & 1.4 & 0 \\
0-3 & 1.2 & 0.3 & 0.05 & 0 \\
0-4 & 37.3 & 72.8 & 1.8  & 0 \\
1-0 & 1.0 & 0.85 & 0.3 & 0 \\
1-1 & 2.6 & 1060.1 & 1.2 & 1.5  \\
1-2 & 90.6 & 1133.8 & 0.7 & 0.2 \\
1-3 & 1.5 & 3.8 & 0.8 & 0 \\
1-4 & 33.7 & 2000.0 & 1.8 & 7.8 \\
2-0 & 0.8 & 0.7 & 0.1 & 0 \\
2-1 & 4.3 & 2000.0 & 2.2 & 11.6 \\
2-2 & 149.7 & 2011.1 & 0.8 & 14.8 \\
2-3 & 1.2 & 12.0 & 0.2 & 0 \\
2-4 & 39.2 & 2001.0 & 1.9 & 11.2 \\
3-0 & 1.1 & 2.4 & 0.3 & 0 \\
3-1 & 3.9 & 2001.0 & 1.7 & 5.6 \\
3-2 & 112.6 & 2015.1 & 0.8 & 4.9 \\
3-3 & 1.6 & 145.3 & 0.2 & 0 \\
3-4 & 33.1 & 2000.5 & 1.2 & 59.9 \\
\hline
mean & 31.0 & 834.6 & 1.0 & 5.9 \\
\hline 
\end{tabular}
\caption{Results i640 class} 
\label{table:resi640} 
\end{table}

\begin{table}
\footnotesize\begin{tabular}{l........}
\hline                       
Name &\multicolumn{2}{c}{MS(-5)} & \multicolumn{2}{c}{MS(-3)} &
\multicolumn{2}{c}{LNDRC} & \multicolumn{2}{c}{DHEA}\\
& \multicolumn{1}{c}{gap(\%)} & \multicolumn{1}{c}{time(s)} &
\multicolumn{1}{c}{gap(\%)} & \multicolumn{1}{c}{time(s)} &
\multicolumn{1}{c}{gap(\%)} & \multicolumn{1}{c}{time(s)} &
\multicolumn{1}{c}{gap(\%)} & \multicolumn{1}{c}{time(s)}\\
\hline                  
6p & 2.2 & 3.5 & 2.6 & 0.6 & 4.2 & 0.5 & 2.2 & 21.3 \\
6u & 1.5 & 6.4 & 4.3 & 0.7 & 4.3 & 0.5 & 1.5 & 0.4 \\
7p & 2.3 & 90.2 & 3.9 & 1.7 & 7.7 & 1.5 & 2.3 & 6000.3 \\
7u & 2.2 & 134.1 & 2.2 & 1.8 & 3.6 & 1.2 & 2.2 & 596.4 \\
8p & 2.4 & 255.5 & 3.4 & 3.8 & 7.1 & 5.2 & 2.3 & 6004.2 \\
8u & 1.8 & 351.1 & 3.3 & 4.9 & 7.5 & 4.1 & 3.3 & 6000.9 \\
9p & 1.8 & 555.6 & 2.3 & 10.8 & 8.6& 16.1 & 22.1 & 6000.0 \\
9u & 1.9 & 775.8 & 3.3 & 11.1 & 6.2 & 13.1 & \multicolumn{1}{c}{Not Found} &
6000.4 \\
10p & 1.7 & 1761.9 & 1.7 & 28.0 & 10.4 & 114.4 & 31.3 & 6000.5 \\
10u & 2.7 & 2468.4 & 2.7 & 32.2 & 7.7 & 59.8 & \multicolumn{1}{c}{Not Found} &
6000.6 \\
11p & 1.5 & 972.3 &  1.6 & 49.3 & 11.6 & 630.0 & \multicolumn{1}{c}{Not Found} &
6003.1 \\
11u & 2.2 & 5632.8 & 2.6 & 71.9 & 9.0 & 360.6 &  \multicolumn{1}{c}{Not Found} &
6001.5 \\
12p & 1.5 & 4970.8 & 1.6 & 121.4 & 11.3 & 3507.7 & \multicolumn{1}{c}{Not
Found}&  6009.8 \\
12u & 2.0 & 4766.7 & 2.4 & 174.1 & 10.0 & 1915.7 & \multicolumn{1}{c}{Not
Found}& 6002.3 \\ \hline
mean & 2.0 & 1624.7 & 2.7 & 36.6 & 7.8 & 473.6 & - & 4760.1 \\
\hline 
\end{tabular}
\caption{Results H class} 
\label{table:resH} 
\end{table}
\section{Post-processing and optimality}

For this section we will assume unbounded depth $D$. Results are not easily
generalizable to the bounded-$D$ case. Results in this section apply to the
non-reinforced MS equations ($\gamma_t = 0$). The results here are based in 
construction of certain trees associated with the original graph and in the 
fact that MS/BP equations are always exact and have a unique solution on
trees\cite{mezard-montanari}.

\begin{defn}
Let $\left\{ \psi_{ij}\right\}$ be a MS fixed-point
(\ref{eq:ms1})-(\ref{eq:ms2}), and
let $\mathbf{d},\mathbf{p}$ be the decisional variables associated
with this fixed point, i.e.
$\left(d_{i}^{*},p_{i}^{*}\right)=\arg\max\psi_{i}\left(d_{i},p_{i}\right)$
for the physical field $\psi_{i}$ from (\ref{eq:marginal}). We will assume this
maximum to be non degenerate. We will employ the \emph{induced}
\emph{subgraph} $S^{*}$=$\left(V^{*},E^{*}\right)$ defined by $V^{*}=\left\{
i\in V:p_{i}^{*}\neq *\right\}\cup\{r\} $
and $E^{*}=\left\{ \left(i,p_{i}^{*}\right):i\in V,p_{i}^{*}\in V\right\} $.
The cost of this subgraph is
$\mathcal{H}\left(S^{*}\right)=\mathcal{H}\left(\mathbf{p}\right)=\sum_{i\in
V}c_{ip_{i}^{*}}$.
\end{defn}

The following optimality property of the MS-induced solution will be proven in the appendix. 
\begin{thm}
\label{thm:MS-exactness}Given a MS fixed point $\left\{ \psi_{ij}\right\} $
on $G$ (unbounded D) with induced subgraph $S^{*}=\left(V^{*},E^{*}\right)$ and
any subtree $S'=\left(V',E'\right)\subseteq G$ with $V'\subseteq V^{*}$,
then $\mathcal{H}\left(S^{*}\right)\leq\mathcal{H}\left(S'\right)$\end{thm}
This result has an easy generalization to loopy subraphs:
\begin{cor}
With $S^*$ as in Theorem \ref{thm:MS-exactness}, given any connected subgraph $S'=\left(V',E'\right)\subseteq G$ with
$V'\subseteq V^{*}$, then
$\mathcal{H}\left(S^{*}\right)\leq\mathcal{H}\left(S'\right)$. \end{cor}
\begin{proof}
Apply Theorem (\ref{thm:MS-exactness}) to a spanning tree of $S'$.
\end{proof}
This trivially implies the following result of optimality of the MS solution in
a particular case:
\begin{cor}
With $S^*=(V^*,E^*)$ as in Theorem \ref{thm:MS-exactness}, if $V^{*}=V$ then
$\mathcal{H}\left(S^{*}\right)=PCST(G)$
\end{cor}
In \cite{Johnson2000}, the MGW algorithm includes two additional methods
to obtain a better PCST solution: StrongPrune and Minimum Spanning
Tree maintaining the same vertex set. Both methods give a substantial
improvement boost to the MGW candidate computed in the first phase.
A natural question may arise, does any of these two methods may help
to improve the solution of MS? The answer is negative in both
cases, and it is a trivial consequence of Theorem (\ref{thm:MS-exactness}).

\begin{cor}
$MST\left(V^{*},E\cap\left(V^{*}\times
V^{*}\right)\right)=\mathcal{H}\left(S^{*}\right)$\end{cor}
\begin{proof}
The minimum spanning tree of $\left(V^{*},E\cap\left(V^{*}\times
V^{*}\right)\right)$
satisfies the hypothesis of Theorem (\ref{thm:MS-exactness}), so
$\mathcal{H}\left(S^{*}\right)\leq MST\left(V^{*},E\cap\left(V^{*}\times
V^{*}\right)\right)$.
The converse inequality is trivially true due to the optimality of
the MST. 
\end{proof}

\begin{cor}
$\mathcal{H}\left(\mbox{StrongPrune}\left(S^{*}\right)\right)=\mathcal{H}
\left(S^{*}\right)$\end{cor}
\begin{proof}
This is a consequence of the fact that
$V\left(\mbox{StrongPrune}\left(S^{*}\right)\right)\subseteq
V\left(S^{*}\right)=V^{*}$
and thus Theorem (\ref{thm:MS-exactness}) applies, implying
$\mathcal{H}\left(S^{*}\right)\leq\mathcal{H}\left(\mbox{StrongPrune}\left(S^{*}
\right)\right)$. The opposite inequality
$\mathcal{H}\left(\mbox{StrongPrune}\left(F\right)\right)\leq\mathcal{H}
\left(F\right)$ was proved in \cite{Johnson2000}.
\end{proof}

\section{Discussion}

In this work we compared MSGSTEINER, an algorithm inspired in the Cavity Theory 
of Statistical Physics, with two state-of-the art algorithms for the 
Prize-Collecting Steiner Problem. The Cavity Theory is expected to give asymptotically exact results on many ensembles of random graphs, so we expected it to give better performance for large instances. The comparison was performed both on 
randomly-generated graphs and existing benchmarks. We observed that 
MSGSTEINER finds better costs in significantly smaller times for many of 
the instances analyzed, and that this difference in time and quality grew with the size of the instances and their solution. We find these results encouraging in views of future applications to problems in biology in which optimization of networks with millions of nodes may be necessary, in particular given the conceptual simplicity of the scheme behind MSGSTEINER (a simple fixed-point iteration).
Additionally, we showed some optimality properties of the Max-Sum (the equations behind MSGSTEINER) fixed points for the unbounded depth case: optimality in some limit cases, and optimality in the general case under the two forms of post-processing present in the MGW algorithm.

\section{Acknowledgements}

Work  supported by  EU Grants No. 267915 and  265496. Work partially supported by the GDRE 224 GREFI-MEFI-CNRS-INdAM.

\appendix
\section{Post-processing and optimality proofs}
Before tackling the proof of the Theorem \ref{thm:MS-exactness}, we will need the following definitions and a technical result.
\begin{defn}
(Computation tree) The computation tree is a cover of the graph
$G$, in the following sense: it is an (infinite) tree $\mathcal{T}_{G}$
along with an application $\pi:\mathcal{T}_{G}\to G$ that satisfies
(\emph{a}) $\pi$ is suryective and (\emph{b}) $\pi_{|i\cup\partial
i}:i\cup\partial i\to\pi\left(i\cup\partial i\right)$
is a graph isomorphism for every $i\in\mathcal{T}_{G}$. It can be
explicitly constructed as the graph of non-backtracking paths in
$G$ starting on a given node $v_{0}$, with two paths being connected
iff the longest one is identical to the other except for an additional
final node (and edge). Up to graph isomorphisms, this tree does not depend 
on the choice $v_0$.

The (finite) tree $\mathcal{T}_{G}\left(t,v_{0}\right)$ is defined
by the radius $t$ ball centered around $v_{o}$ in $\mathcal{T}_{G}$.
Alternatively, it can be directly constructed as the graph of non-backtracking
paths of length $t$ starting on $v_{0}$, with two paths being connected
iff the longest one is identical to the other except for an additional
final node (and edge). Clearly the finite computation tree depends 
strongly on the choice of $v_0$

For both computation trees, edge weights (and node prizes) will be lifted
(transported) naturally as $c_{ij}=c_{\pi\left(i\right)\pi\left(j\right)}$.

Lifting edge constraints by $g_{ij}=g_{\pi(i)\pi(j)}$ defines a $(R,D)$-PCSF
problem with $R=\pi^{-1}(\{r\})$ on $\mathcal{T}_G$. On
$\mathcal{T}_{G}\left(t,v_{0}\right)$ instead, it gives a slightly relaxed
$(R,D)$-PCSF problem in which leaf nodes can point to neighbors in $G$ that are
not present in $\mathcal{T}_G$.
For convenience, let us extend $\pi$ by setting $\pi\left(*\right)=*$.

\end{defn}
\begin{rem}
As $\mathcal{T}_{G}(t,v)$ is a tree, the MS equations are exact
and have a unique fixed point in $\mathcal{T}_{G}\left(t,v\right)$\cite{mezard-montanari}.\end{rem}
\begin{lem}
\label{lem:lifting-messages}Any MS fixed point in a graph $G$ can
be naturally lifted to a MS fixed point in $\mathcal{T}_{G}$. Moreover,
any MS fixed point can be naturally lifted to a MS fixed point over
a slightly modified $\mathcal{T}_{G}(t,v)$ with extra cost terms only on leaves.\label{lemma:lifting}
\end{lem}
\begin{proof}
As MS equations are local and the two graphs are locally isomorphic,
given a fixed point $\left\{ \psi_{ij}\right\} _{\left(i,j\right)\in E}$,
the messages $\Psi_{ij}=\psi_{\pi\left(i\right)\pi\left(j\right)}$
satisfy the fixed point equations on $\mathcal{T}_{G}$. On
$\mathcal{T}_{G}(t,v)$
the MS equations are satisfied everywhere except possibly on leaf
nodes (where the graphs are not locally isomorphic). Given a leaf
$i$ attached with edge $\left(i,j\right)$, add an energy term
$-E_{i}\left(d_{i}p_{i}\right)=\psi_{\pi\left(i\right)\pi\left(j\right)}\left(d_
{i},\pi\left(p_{i}\right)\right)$.
Now MS equations are satisfied everywhere on for this modified cost function.
\end{proof}

Now we proceed to prove Theorem \ref{thm:MS-exactness}
\subsection{Proof of Theorem \ref{thm:MS-exactness}}
\begin{proof}
Assume $S'$ oriented towards the root node $r$, i.e. defining a
parenthood vector $\left(p_{i}'\right)_{i\in V'}$, such that
$E'=\{(i,p_{i}^{'}):i\in V'\setminus\left\{ r\right\} \}$.
Consider the subgraph $S=\left(V_{S},E_{S}\right)$ of
$\mathcal{T}_{G}(N+1,r)$
induced by $S^{*}$, i.e. defined by $V_{S}=\left\{ v:\pi\left(v\right)\in
V^{*}\right\} $,
$E_{S}=\{\left(i,j\right):\left(\pi\left(i\right),\pi\left(j\right)\right)\in
E^{*}\}$.

It can be easily proven that the connected component in $S$ of the
root node of $\mathcal{T}_{G}(N+1,r)$ is a tree $S''$ isomorfic
to $S^{*}$ (see \cite{BBZ2008b}). Denote by $\left\{ p^{*}\right\} $ the
decisional variables
induced by $S^{*}$ and by $\left\{ p'\right\} $ the ones induced
by $S'$. The parenthood assignment 
\[
q_{i}=\begin{cases}
p'_{i} & i\in V_{S''}\\
p_{i}^{*} & i\notin V_{S''}
\end{cases}
\]
satisfies $q_{i}\neq*$ if $q_{j}=i$ (as $V'\subseteq V^{*}$) and
so depths $d_i$ can be assigned so as to verify all $g_{ij}$ constraints
in $\mathcal{T}_{G}(N+1,r)$. Now the cost associated with $\mathbf{q}$ is
$\mathcal{H}\left(\mathbf{q}\right)=\sum_{i\in V_{S''}}c_{ip'_{i}}+\sum_{i\notin
V_{S''}}c_{ip^*_{i}}\geq\sum_{i\in\mathcal{T}_{G}(N+1,r)}c_{ip_{i}^{*}}=\sum_{
i\in V_{S''}}c_{ip_{i}^{*}}+\sum_{i\notin V_{S''}}c_{ip{}_{i}^{*}}$
due to the optimality of the MS solution $\mathbf{p}^{*}$ in the computation
tree (this is because MS is always exact on a tree). This implies clearly that
$\mathcal{H}\left(S^{*}\right)\leq\mathcal{H}\left(S'\right)$. \end{proof}

\bibliographystyle{splncs}

\end{document}